\documentclass[runningheads]{llncs}
\usepackage{graphicx}
\usepackage{xcolor}
\usepackage{amsmath,amssymb}
\usepackage{mathtools, stmaryrd}
\newtheorem{thm}{Theorem}
\newtheorem{prop}{Proposition}

\newtheorem{defn}{Definition}
\newtheorem{corol}{Corollary}

\let\doendproof\endproof
\renewcommand\endproof{~\hfill\qed\doendproof}

\newcommand{\LatL}{\mathcal{L}}
\newcommand{\dl}{\delta}
\newcommand{\ep}{\varepsilon}
\newcommand{\gm}{\gamma}

\newcommand{\Gp}{\Gamma^{(p)}}

\newcommand{\Ginf}{\Gamma^{(\infty)}}

\newcommand{\Ccal}{\mathcal{C}}
\newcommand{\Gcal}{\mathcal{G}}
\newcommand{\N}{\mathbb{N}}
\newcommand{\R}{\mathbb{R}}
\newcommand{\Rmax}{\mathbb{R}_{\max}}
\newcommand{\Rmin}{\mathbb{R}_{\min}}
\newcommand{\matRmax}[1]{\mathcal{M}_{#1}}
\newcommand{\setBrace}[2]{\lbrace #1,\dots, #2\rbrace}
\newcommand{\setBracket}[2]{\llbracket #1, #2\rrbracket}

\newcommand{\mul}{\otimes}

\newcommand{\mbf}[1]{\mathbf{#1}}

\newcommand{\doa}{doubly-$0$-astic}
\newcommand{\Doa}[1]{\mathcal{D}_0(#1)}

\begin{document}
\title{Morphological adjunctions represented by matrices in max-plus
  algebra for signal and image processing}
\titlerunning{Morphological adjunctions represented by matrices}
%
\author{{Samy} Blusseau\inst{1} \and
{Santiago} Velasco-Forero\inst{1} \and
{Jes\'us} Angulo\inst{1} \and
{Isabelle} Bloch\inst{2}}
\authorrunning{S. Blusseau et al.}
%
\institute{Mines Paris, PSL University, Centre for Mathematical Morphology,\\77300 Fontainebleau, France \and
Sorbonne Universit\'e, CNRS, LIP6, F-75005 Paris, France}
%
\maketitle              
\begin{abstract}

  In discrete signal and image processing, many dilations and erosions
  can be written as the max-plus and min-plus product of a matrix on a
  vector. Previous studies considered operators on symmetrical,
  unbounded complete lattices, such as Cartesian powers of the
  completed real line. This paper focuses on adjunctions on closed
  hypercubes, which are the complete lattices used in practice to
  represent digital signals and images. We show that this constrains the
  representing matrices to be doubly-$0$-astic and we characterise the
  adjunctions that can be represented by them. A graph interpretation
  of the defined operators naturally arises from the adjacency
  relationship encoded by the matrices, as well as a max-plus spectral
  interpretation.

\end{abstract}

\section{Introduction}
\label{sec:intro}

Like linear filters can be represented by matrices in discrete image
and signal processing, many morphological dilations and erosions can
be seen as applying a matrix product to a vector, but in the minimax
algebra. This is in particular the case for those defined with
structuring functions, either flat or not, local or
non-local~\cite{salembier09study,velasco13nonlocal}, translation
invariant or spatially
variant~\cite{bouaynaya2008theoretical,debayle2006general,lerallut2007amoebas,verdu2011anisotropic}. They
are commonly known to be the vertical-shift-invariant dilations and
erosions~\cite{heijmans1990algebraic,maragos13representations}. While
the matrix point of view is not the most appropriate for the
implementation of these operators, especially translation-invariant
ones, it is a valuable insight for their theoretical understanding. In
particular, it can help predict and control complex behaviours such as
those of iterated operators based on adjunctions with non-flat,
spatially variant and input-adapted structuring
functions~\cite{blusseau18tropical,blusseau2022anisotropic}. Indeed,
it is a flexible and general framework which embraces a very broad
part of morphological literature, and it is supported by the rich
theory of Minimax
algebra~\cite{akian06max-plus-algebra,cuninghame79minimax}.

In the abundant literature on spatially-variant morphological image
processing, only a few approaches explicitly used the matrix
formulation~\cite{blusseau18tropical,blusseau2022anisotropic,salembier09study,velasco13nonlocal},
whereas most contributions were limited to flat structuring elements
and focused on the local effects of the adaptive strategy. On the
theoretical side, the representation of morphological adjunctions by
matrices was studied in a setting that does not directly apply to
digital signal and image processing, as the co-domain is usually an unbounded lattice, stable under any vertical
translation~\cite{maragos13representations}. Although a method was proposed to convert these adjunctions to new ones on bounded lattices~\cite{ronse1990why}, it is not practical and does not allow for the interpretations that are exposed here.

In the present paper, we focus on complete lattices of the type
$[a, b]^n$, where $a$ and $b$ represent the minimal and maximal
possible signal values (typically $a = 0$ and $b = 255$ for 8-bits
images), and $n$ is an integer representing the size of the signal
(typically, the number of pixels of an image, reshaped as a column
vector). This is a theoretical contribution that can be viewed as a
companion paper to previous studies where this framework has been
successfully applied to adaptive anisotropic filtering~\cite{blusseau18tropical,blusseau2022anisotropic}\footnote{An online demo for~\cite{blusseau2022anisotropic} is available:
  \url{https://bit.ly/anisop\_demo}.}.
In Section~\ref{sec:morpho_operators} we introduce the matrix-based
morphological setting and prove simple but fundamental results:
in particular, we characterise the adjunctions that can be represented by matrices and show that these matrices need to be doubly-$0$-astic.
By viewing matrices as encoding adjacency, we provide in Section~\ref{sec:graph-tropical} a graph interpretation of iterated operators and their associated
granulometries. In Section~\ref{sec:morpho_spectral} we draw a link
between these operators and some results on the spectrum of matrices
in the max-plus algebra, before concluding in
Section~\ref{sec:conclusion}.

\section{Matrix-based morphological adjunctions}
\label{sec:morpho_operators}

\subsection{Notations}
\label{sec:notations}

In this paper matrices will be denoted by capital letters, such as
$W$, and their $i$-th row and $j$-th column coefficients by
corresponding indexed lowercase letters~$w_{ij}$. Similarly, vectors
are written as boldface lowercase letters, such as $\mbf{x}$, and
their $i$-th component as $x_i$.
Let $0\leq a < b\in\mathbb{R}^+$ be two non-negative real numbers,
$n\in\N^*$ a positive integer.
The set $\setBrace{1}{n}$ will be denoted by $\llbracket 1,n \rrbracket$.
Let $\LatL = ([a,b]^n, \leq)$ be the
complete lattice equipped with the usual product partial ordering
(Pareto ordering):
$\mbf{x}\leq\mbf{y} \iff x_i \leq y_i, \;\;\forall
i\in\setBracket{1}{n}$. The supremum and infimum on $\LatL$ are induced
by the Pareto ordering: for a family $(\mbf{x}^{(k)})_{k\in K}$ of
$\LatL$, $\bigvee_{k\in K} \mbf{x}^{(k)}$ is the vector $\mbf{y}$
defined by $y_i = \bigvee_{k\in K}x_i^{(k)}$, where $K$ is any index
set. Therefore $\mbf{a} = (a, \dots, a)^T$ and
$\mbf{b} = (b, \dots, b)^T$ are respectively the smallest and largest
elements in $\LatL$. For $\mbf{x}\in\LatL$, we note
$\mbf{x}^c \ \dot{=} \ \mbf{b} - \mbf{x} +\mbf{a}$, and for any
$i\in \setBrace{1}{n}$, $\mbf{e^{(i)}}$ is the ``impulse'' vector in
$\LatL$ such that $e^{(i)}_i = b$ and $e^{(i)}_j = a$ for
$j \neq i$.

We note $\Rmax \ \dot{=}\ \mathbb{R}\cup \lbrace-\infty\rbrace$,
$\Rmin \ \dot{=} \ \mathbb{R}\cup \lbrace +\infty\rbrace$ and $\matRmax{n}$
the set of $n\times n $ square matrices with coefficients in $\Rmax$.
Like $(\Rmax, \vee, +)$, $(\matRmax{n}, \vee, \mul)$ is an idempotent semiring, with the addition $\vee$ and product $\mul$ defined as follows.
For $A, B\in\matRmax{n}$, $A \mul B$ and $A\vee B$  are the $n\times n$ matrices defined respectively by
$(A \mul B)_{ij} = \bigvee_{k=1}^n a_{ik}+ b_{kj} $ and $(A\vee B)_{ij} = a_{ij}\vee b_{ij} = \max(a_{ij}, b_{ij})$, for
$1\leq i,j\leq n$. 
Similarly, for $\mbf{x}\in\Rmax^n$, $A\mul\mbf{x}$ is the vector such that $(A\mul\mbf{x})_i = \bigvee_{j=1}^n a_{ij} + x_j $.
Note that
$\vee$ and $\mul$ are associative and $\mul$ is distributive over
$\vee$. Finally, the product of a scalar $\lambda\in\Rmax$ by a vector $\mbf{x}\in\Rmax^n$ is $\lambda\mul\mbf{x} \ \dot{=} \ \lambda + \mbf{x}$, the vector in $\Rmax$ such that
$(\lambda \mul \mbf{x})_i = \lambda +
x_i$. 
In~\cite{cuninghame79minimax} and~\cite{velasco13nonlocal}, special
subsets of $\matRmax{n}$ are introduced, that we will show to be
essential to represent morphological adjunctions on $\LatL$.

\begin{defn}[0-asticity~\cite{cuninghame79minimax}]
  \label{def:asticity}
  A matrix $W\in\matRmax{n}$ is said \textbf{\emph{row-0-astic}}
  if for any $1\leq i \leq n, \ \bigvee_{j=1}^n w_{ij} = 0$. 
  Similarly, it is said \textbf{\emph{column-0-astic}} if the
  supremum of each column is $0$, and \textbf{\emph{doubly-0-astic}}
  if the matrix is both row-0-astic and column-0-astic.
  Finally, $W$ is simply said \textbf{\emph{$0$-astic}} if
  $\ \bigvee_{1\leq i,j \leq n} w_{ij} = 0.$
\end{defn}
A special kind of {\doa} matrices are those with zeros on the diagonal
and non-positive coefficients elsewhere.
\begin{defn}[CMW matrices~\cite{velasco13nonlocal}]
  \label{def:cmw}
  A matrix $W\in\matRmax{n}$ is called a Conservative Morphological
  Weights (CMW) matrix if $\forall i, j \in\setBracket{1}{n}$,
  $w_{ij}\leq 0$ and $w_{ii} = 0$.
\end{defn}

We now introduce the morphological framework on $\LatL$, based on the
max-plus algebra product between matrices and vectors.

\subsection{Dilations}
\label{sec:dilations}

For $W\in\matRmax{n}$, we consider the
function $\dl_W$ from $\LatL$ to $\Rmax^n$ such that
\begin{equation}
  \label{eq:delta_W}
\forall\mbf{x}\in\LatL,\;\;\;\;\dl_W(\mbf{x}) = W\mul\mbf{x} = \left(\bigvee_{1\leq j \leq n} \lbrace w_{ij} +  x_j \rbrace\right)_{1\leq i \leq n} .
\end{equation}
In the processing of digital data such as images we usually want the
input to be comparable with the output. Hence, we will constrain $W$
such that $\dl_W(\LatL) \subseteq \LatL$. This has the following
consequences:
$$
\dl_W(\mbf{b})\leq \mbf{b} \Rightarrow \forall i\in\setBracket{1}{n},
\;\;b + (\bigvee_{j=1}^n w_{ij}) \leq b \Rightarrow \forall i\in\setBracket{1}{n},
\;\; \bigvee_{j=1}^n w_{ij} \leq 0
$$
since $b>-\infty$. Similarly,
$ \dl_W(\mbf{a})\geq \mbf{a} \Rightarrow \forall i\in\setBracket{1}{n},
\;\; \bigvee_{j=1}^n w_{ij} \geq 0.$
Hence a necessary condition to have $\dl_W(\LatL) \subseteq \LatL$ is that
$W$ be row-$0$-astic (Def.~\ref{def:asticity}). Conversely, the
row-$0$-asticity for $W$ implies that $\dl_W(\mbf{a}) =
\mbf{a}$ and $\dl_W(\mbf{b}) = \mbf{b}$, and therefore
that $\dl_W(\LatL) \subseteq \LatL$ by increasingness of $\dl_W$. This leads to the following result.

\begin{prop}
  \label{prop:dilation-roa}
  Let $W\in\matRmax{n}$ and $\dl_W$ be the function defined by
  (\ref{eq:delta_W}). Then $\dl_W$ is a dilation mapping $\LatL$ to
  $\LatL$ if and only if $W$ is row-$0$-astic.
\end{prop}
\begin{proof}
  If $\dl_W$ is a dilation mapping $\LatL$ to $\LatL$, then
  $\dl_W(\LatL) \subseteq \LatL$ which, as we showed, implies that $W$
  is row-$0$-astic. Conversely, we saw that a row-$0$-astic $W$
  implies $\dl_W(\LatL) \subseteq \LatL$. Therefore, we only have to
  verify that $\dl_W$ is a dilation, or equivalently that it commutes
  with the supremum. This is straightforward from the definition of
  $W\mul\mbf{x}$.
\end{proof}

\subsection{Erosions and adjunctions}
\label{sec:erosions-adjunctions}

Now we suppose that $W\in\matRmax{n}$ is row-$0$-astic, hence $\dl_W$
is a dilation from $\LatL$ to $\LatL$, and we are interested in its
adjoint erosion $\alpha_W$ defined for any $\mbf{y}\in\LatL$ by
$\alpha_W(\mbf{y}) = \bigvee E_{\mbf{y}}$ where
$ E_{\mbf{y}} = \lbrace \mbf{x}\in\LatL, \dl_W(\mbf{x}) \leq
\mbf{y}\rbrace.$ Let us denote by $\ep_W$ the function from $\LatL$ to
$\Rmin^n$ such that for any $\mbf{y}\in\LatL$
\begin{equation}
  \label{eq:epsilon_W}
  \ep_W(\mbf{y}) 
  =  \big(\dl_{W^T}(\mbf{y}^c)\big)^c = \big(W^T\mul\mbf{y}^c\big)^c = \left(\bigwedge_{1\leq j \leq n} \lbrace y_j - w_{ji}\rbrace\right)_{1\leq i \leq n}.
\end{equation}
Then we can check that
$\forall \mbf{y}\in\LatL,\;\; \alpha_W(\mbf{y}) = \ep_{W}(\mbf{y})
\wedge \mbf{b}$. Indeed, from \eqref{eq:delta_W} we see that for any
$\mbf{x}, \mbf{y}\in \LatL, \ \dl_W(\mbf{x})\leq \mbf{y} \iff
\mbf{x}\leq \ep_W(\mbf{y}).$ Therefore, since
$\delta_W(\mbf{a}) = \mbf{a} \leq \mbf{y}$ we get
$\mbf{a}\leq \ep_{W}(\mbf{y})$, which implies
$\ep_{W}(\mbf{y}) \wedge \mbf{b}\in\LatL$; furthermore,
$\ep_{W}(\mbf{y})\wedge\mbf{b}\leq \ep_{W}(\mbf{y})$ so
$\ep_{W}(\mbf{y})\wedge\mbf{b}\in E_{\mbf{y}}$; finally, as both
$\ep_{W}(\mbf{y})$ and $\mbf{b}$ are upper-bounds of $E_{\mbf{y}}$, so
is $\ep_{W}(\mbf{y})\wedge\mbf{b}$. Hence,
$\ep_{W}(\mbf{y})\wedge\mbf{b} = \bigvee E_{\mbf{y}} =
\alpha_W(\mbf{y})$. By a similar reasoning as in
Section~\ref{sec:dilations}, we get the following result.

\begin{prop}
\label{prop:erosion-coa}
Let $W\in\matRmax{n}$ and $\ep_W$ be the function defined by
(\ref{eq:epsilon_W}). Then $\ep_W$ is an erosion mapping $\LatL$ to
$\LatL$ if and only if $W$ is column-$0$-astic.
\end{prop}
If $W$ is also row-$0$-astic, then $\ep_W = \alpha_W$ is the adjoint
of $\dl_W$, as stated next.
\begin{prop}
  \label{prop:adjunction-doa}
  Let $W\in\matRmax{n}$ and $\dl_W$ and $\ep_W$ be the functions
  defined by (\ref{eq:delta_W}) and (\ref{eq:epsilon_W}),
  respectively. Then ($\ep_W, \dl_W$) is an adjunction on $\LatL$ if
  and only if $W$ is \doa. Furthermore, ($\ep_W, \dl_W$) is an
  adjunction on $\LatL$ with $\dl_W$ extensive (and $\ep_W$
  anti-extensive) if and only if $W$ is a CMW matrix.
\end{prop}
\begin{proof}
  Most of the points have already been addressed above or are
  straightforward from Proposition~\ref{prop:dilation-roa}. To see
  that $\dl_W$ extensive implies $w_{ii} = 0$ for all $i$, just remark
  that $w_{ii} < 0$ would imply $\dl_W(\mbf{e^{(i)}})_i< b = e^{(i)}_i$.
\end{proof}

\subsection{Generality of $(\ep_W, \ \dl_W)$}
\label{sec:generality}

The dilation $\dl_W$, already introduced
in~\cite{blusseau18tropical,maragos13representations,velasco-forero15nonlinear},
can be viewed as a generalisation of the non-local and adaptive
mathematical morphology~\cite{salembier09study,velasco13nonlocal} on
signals and images. Each column $W_{:,j}$ of $W$ represents the
structuring function corresponding to pixel (or instant) $j$.

As pointed
out in~\cite{heijmans1990algebraic,maragos13representations}, the dilations
that can be written as matrix-based max-plus products like
Eq.~\eqref{eq:delta_W} are the shift (or \emph{vertical-translation})
invariant ones. However the result stated
in~\cite{heijmans1990algebraic,maragos13representations} does not
directly apply to our setting where the lattice $\LatL$ is
different from the lattice of scalars which define vertical
translation of signal values, usually
$\R\cup\lbrace -\infty, +\infty\rbrace$. Still, the same idea holds
here with some adaptation, as stated in the next proposition.

\begin{prop}
  \label{prop:generality}
  Let $\dl : \LatL \to \LatL$ be a dilation. Then there exists
  $W \in\matRmax{n}$ such that $\dl = \dl_W$ if and only if
  \begin{equation}
    \label{eq:shift_invariance}
    \forall \lambda \leq 0, \forall \mbf{x}\in\LatL,\;\; \dl\big( (\lambda + \mbf{x})\vee \mbf{a} \big) = \big(\lambda + \dl(\mbf{x})\big)\vee \mbf{a}.
  \end{equation}
  In that case, the matrix $W$ whose $j$-th column is
  $W_{:,j}= \dl(\mbf{e^{(j)}})-\mbf{b}$ for $1\leq j \leq n$, is such
  a representing matrix.
\end{prop}
We see that this class of dilations is very broad and covers the most
commonly used in morphological image and signal processing: dilations
based on structuring functions, possibly non-local, varying in space
and non-flat.
\begin{proof}[Proposition~\ref{prop:generality}]
  If $\dl = \dl_W$ for some $W \in\matRmax{n}$, then it is
  straightforward to check that $\dl$ verifies
  Eq.~\eqref{eq:shift_invariance}.

  Conversely, suppose $\dl$ verifies
  Eq.~\eqref{eq:shift_invariance}. Then we first remark that
  $\dl(\mbf{b}) = \mbf{b}$. Indeed, on the one hand,
  $\dl(\mbf{a}) = \mbf{a}$ as $\mbf{a}=\bigwedge\LatL$ and $\dl$ is a
  dilation mapping $\LatL$ to $\LatL$. On the other hand,
  $\dl(\mbf{a}) = \dl\big( (a-b) + \mbf{b}\big) = (a-b) +
  \dl(\mbf{b})$ by Eq.~\eqref{eq:shift_invariance}. Hence
  $\mbf{a} = (a-b) + \dl(\mbf{b})$ which means that
  $\dl(\mbf{b}) = \mbf{b}$.\\
  As a consequence: for any $i\in\setBracket{1}{n}$, there is a
  $j_i\in\setBracket{1}{n}$ such that $\dl(\mbf{e^{(j_i)}})_i = b$. This
  is simply because $\mbf{b} = \bigvee_{1\leq j \leq n} \mbf{e^{(j)}}$
  so
  $\mbf{b}= \dl\big(\bigvee_{1\leq j \leq n}\mbf{e^{(j)}}\big) =
  \bigvee_{1\leq j \leq n} \dl(\mbf{e^{(j)}})$, which means that, for
  any $i$, $b = \bigvee_{1\leq j \leq n} \dl(\mbf{e^{(j)}})_i$ and
  finally that $\dl(\mbf{e^{(j_i)}})_i = b$ for some $j_i$, as the
  supremum is reached here.\\
  Now, let $\mbf{x}\in\LatL$. Then it can be decomposed as
  $\mbf{x} = \bigvee_{1\leq j \leq n} \big[(\lambda_j +
  \mbf{e^{(j)}})\vee \mbf{a}\big]$ with $\lambda_j = x_j - b \leq
  0$. Hence, as $\dl$ is a dilation verifying
  Eq.~\eqref{eq:shift_invariance}, we get
  $\dl(\mbf{x}) = \bigvee_{1\leq j \leq n} \big[\big(\lambda_j +
  \dl(\mbf{e^{(j)}})\big)\vee \mbf{a}\big]$. We now use the result
  stated just above: for any $i\in\setBracket{1}{n}$ there is a
  $j_i\in\setBracket{1}{n}$ such that
  $\lambda_{j_i} + \dl(\mbf{e^{(j_i)}})_i= x_{j_i} - b + b = x_{j_i}
  \geq a$. Therefore,
  $\bigvee_{1\leq j \leq n} \big[\big(\lambda_j +
  \dl(\mbf{e^{(j)}})\big)\vee \mbf{a}\big] = \bigvee_{1\leq j \leq n}
  \lambda_j + \dl(\mbf{e^{(j)}})$ from which we finally get
  \begin{equation}
    \label{eq:delta_x_general}
    \dl(\mbf{x}) = \bigvee_{1\leq j \leq n} \lambda_j + \dl(\mbf{e^{(j)}}) = \bigvee_{1\leq j \leq n} (x_j-b) + \dl(\mbf{e^{(j)}}) = \bigvee_{1\leq j \leq n} x_j + [\dl(\mbf{e^{(j)}}) - \mbf{b}]
  \end{equation}
  which is exactly $W\mul\mbf{x}$ for $W$ the matrix with columns
  $W_{:,j}= \dl(\mbf{e^{(j)}})-\mbf{b}$ for $1\leq j \leq n$.
\end{proof}
  Note that the dual of Proposition~\ref{prop:generality} obviously holds: the
  erosions $\ep:\LatL\to\LatL$ which can be written as $\ep_W$ for some
  $W\in\matRmax{n}$ are those for which
  \begin{equation}
    \label{eq:shift_invariance2}
    \forall \lambda \geq 0, \forall \mbf{x}\in\LatL,\;\; \ep\big( (\lambda + \mbf{x})\wedge \mbf{b} \big) = \big(\lambda + \ep(\mbf{x})\big)\wedge \mbf{b}.
  \end{equation}
  To show this it is sufficient to see that $\ep$ verifies
  \eqref{eq:shift_invariance2} if and only if the dilation
  $\dl = \ep(\cdot^c)^c$ verifies \eqref{eq:shift_invariance}, and
  recall that $\dl_W(\cdot^c)^c = \ep_{W^T}(\cdot)$.

\subsection{Equivalent dilations and erosions}
\label{sec:eqivalent_operators}

In Proposition~\ref{prop:generality} we exhibited one possible matrix
$W\in\matRmax{n}$ that represents a dilation, but this matrix is not
unique. In this section we characterise the set of such matrices and
show that it is a complete lattice.

Since we are interested in adjunctions $(\ep_W, \dl_W)$, following
Proposition~\ref{prop:adjunction-doa} we focus on the set of matrices
in $\matRmax{n}$ that are \doa, which we denote by $\Doa{n}$. Let the
equivalence relation defined for any
two matrices $A, B\in\Doa{n}$ by
\begin{equation}
  \label{eq:equivalence-relation-doa}
  A \sim B \iff \dl_A = \dl_B \iff \forall\mbf{x}\in \LatL,\ \dl_A(\mbf{x}) = \dl_B(\mbf{x})
\end{equation}
and note
$\mathcal{C}_W = \left\lbrace M\in\Doa{n}, \ M\sim W \right\rbrace$
the equivalence class of any $W\in\Doa{n}$.
We provide an easy characterisation of $\mathcal{C}_W$ that will show
useful in numerical computations of the morphological operators
defined earlier. For any $u\in\Rmax$ let $I_{u}$ denote the matrix in
$\matRmax{n}$ whose coefficients are all equal to $u$. Then two
equivalent matrices are characterised as follows.
\begin{prop}
  \label{prop:characterise-equivalence-class}
  Let $M, W\in\Doa{n}$. Then
  \begin{equation}
    \label{eq:characterise-equivalence-class}
    \begin{array}{lllll}
      M\in\mathcal{C}_W &\iff 
      & M\vee I_{a-b} = W\vee I_{a-b} & \iff 
      & 
        \left\lbrace
        \begin{array}{cl}
          m_{ij} = w_{ij} & \text{ if }  w_{ij} > a-b\\
          m_{ij} \leq a-b & \text{ otherwise. }
        \end{array}\right.
      \\
    \end{array}
  \end{equation}
\end{prop}
This means that if $W$ has coefficients not larger than $a-b$, these
can be set to any value not larger than $a-b$, including $-\infty$,
and can therefore be ignored in the computation of $\dl_W(\mbf{x})$.
\begin{proof}[Proposition~\ref{prop:characterise-equivalence-class}]
  The second equivalence is just a matter of writing, so we prove the
  first one.
  Let us first notice that for any
  $\mbf{x}\in\LatL, \ I_{(a-b)}\mul \mbf{x} \leq
  \mbf{a}$. Therefore
  $
  \forall \mbf{x}\in \LatL, \;\; (W\vee I_{(a-b)})\mul \mbf{x} =
  (W\mul\mbf{x}) \vee (I_{(a-b)}\mul \mbf{x}) = W\mul\mbf{x}
  $,
  since $W\mul\mbf{x}\geq \mbf{a}$, and this holds for $M$ too.
  Hence, if $M\vee I_{a-b} = W\vee I_{a-b}$, then for any
  $\mbf{x}\in \LatL, \ W\mul\mbf{x} = (W\vee I_{a-b})\mul\mbf{x} =
  (M\vee I_{a-b})\mul\mbf{x} = M\mul\mbf{x}, $ which means
  $M\in\mathcal{C}_W$.

  Conversely, suppose that $M\sim W$ and that $w_{i_0 j_0} > a-b$ for some
  $i_0, j_0 \in\setBracket{1}{n}$. Let
  $\mbf{x} = \mbf{e^{(j_0)}}\in \LatL$, \emph{i.e.} $x_{j_0}= b$ and
  $x_j=a \;\; \forall j\neq j_0$. The $0$-asticity of $W$ and $M$
  implies $(W\mul\mbf{x})_{i_0} = b+w_{i_0 j_0}$ and
  $(M\mul\mbf{x})_{i_0} = b+m_{i_0 j_0}$, hence
  $m_{i_0j_0} = w_{i_0j_0}$. We have just shown that
  $\forall i, j \in\setBracket{1}{n},\ \left(w_{ij} > a-b \Rightarrow
    w_{ij} = m_{ij}\right)$ and by symmetry of the equivalence
  relation $ \left(m_{ij} > a-b \Rightarrow w_{ij} = m_{ij}\right)$,
  which combined yields $\max(m_{ij},\ a-b) = \max(w_{ij},\ a-b)$. So
  finally $M\sim W \Rightarrow M\vee I_{a-b} = W\vee I_{a-b}$.
\end{proof}
While it is clear that if $A, B\in \mathcal{C}_W$ then
$A\vee B \in \mathcal{C}_W$, the characterisation in
Proposition~\ref{prop:characterise-equivalence-class} shows that
$\mathcal{C}_W$ is also closed under infimum, that is:
$A \wedge B \in \mathcal{C}_W$. This has the following straightforward
consequence.
\begin{prop}
  \label{prop:L_W-complete-lattice}
  Let $W\in\Doa{n}$ and $\leq$ the partial ordering on $\mathcal{C}_W$
  defined by $A\leq B \iff A\vee B = B \iff a_{ij}\leq b_{ij} \ \forall
  i, j\in\setBracket{1}{n}$. Then
  \begin{itemize}
  \item $(\mathcal{C}_W, \leq)$ is a complete lattice (with
    coefficient-wise supremum and infimum);
  \item Its greatest element is $\overline{W} = W\vee I_{a-b}$;
  \item Its smallest element is $\underline{W}$, defined by
    $\underline{w}_{ij} =\left\lbrace
    \begin{array}{cl}
      w_{ij} & \text{ if } w_{ij} > a-b\\
      -\infty & \text{ otherwise.}
    \end{array}\right.
$
  \end{itemize}
\end{prop}

\subsection{Iterated operators and granulometries}
\label{sec:iterated_operators}

In this section, given $W\in\Doa{n}$ and $p\in\N^*$, we focus on the
iterated dilations and erosions $\dl^p_W$ and $\ep^p_W$, as well as
their sup and inf integrations, that we note respectively
$D^{[p]}_W \ \dot{=} \ \bigvee_{k=1}^p\dl^k_{W}$ and
$E^{[p]}_W \ \dot{=} \ \bigwedge_{k=1}^p\ep^k_{W}$. One can easily
check that both $(\ep^p_W, \dl^p_W)$ and $(E^{[p]}_W, D^{[p]}_W)$ are
adjunctions. We note respectively
$\gm^{[p]}_W \ \dot{=} \dl^p_W\ep^p_W$ and
$G^{[p]}_W \ \dot{=} D^{[p]}_WE^{[p]}_W$ their corresponding openings.

Note that if $\dl_W$ is extensive, or equivalently if $W$ is a CMW
matrix (Prop.~\ref{prop:adjunction-doa}), then these adjunctions are
equal: $(\ep^p_{W}, \dl^p_{W}) = (E^{[p]}_{W}, D^{[p]}_{W})$. As this
is not true in general, both adjunctions are worth studying. In
particular, we shall examine whether $(\gm^{[p]}_W)_{p\in\N^*}$ and $(G^{[p]}_W)_{p\in\N^*}$ define granulometries, that is to say families of openings that are decreasing with $p$. The answer is yes and it is a general result that does not depend on the
representation of the adjunction.

\begin{prop}
  \label{prop:granulometries}
  Let $(\ep, \ \dl)$ be an adjunction on a complete lattice. For any
  integer $p\in\N^*$, let us note $\gamma_p = \dl^p\ep^p$ and
  $G_p = D_pE_p$ the openings associated to the adjunctions
  $(\ep^p, \ \dl^p)$ and
  $\left(E_p = \bigwedge_{1\leq k \leq p}\ep^k, D_p = \bigvee_{1\leq k
      \leq p}\dl^k\right)$, respectively. Then $(\gamma_p)_{p\in\N^*}$
  and $(G_p)_{p\in\N^*}$ are granulometries.
\end{prop}
\begin{proof}
  We first show that the family of openings $(\gamma_p)_{p\geq 1}$
  decreases with $p$, hence a granulometry. This is straightforward by
  writing
  $\gamma_{p+1} = \dl^{p+1}\ep^{p+1} = \dl^{p}\gamma_1\ep^{p} \leq
  \dl^{p}\ep^{p} = \gamma_p.$ Secondly, regarding $(G_p)_{p\geq 1}$, we show $G_{p+1} \leq G_{p}$ by
  proving that $G_pG_{p+1} = G_{p+1}$. We obtain this by remarking
  that $D_{p+1} = D_p(id\bigvee \dl)$, which makes it an invariant of
  $G_p$:
  $G_p D_{p+1} = D_pE_pD_p(id\bigvee\dl) = D_p(id\bigvee\dl) =
  D_{p+1}$. Then we can conclude
  $G_pG_{p+1} = G_p D_{p+1}E_{p+1} = D_{p+1}E_{p+1} = G_{p+1}$.
\end{proof}
To conclude this section, let us write $\dl^p_W$, $\ep^p_W$,
$D^{[p]}_W $ and $E^{[p]}_W $ as dilations and erosions represented by
one suitable {\doa} matrix. This will help in their graph
interpretation of the next section. The associativity of $\mul$ yields
$\forall \mbf{x}\in\LatL, \ \dl^p_{W}(\mbf{x})=W\mul\dots\mul
W\mul\mbf{x} = W^p\mul\mbf{x}$, therefore $\dl^p_{W} = \dl_{W^p}$. We
obtain similarly $\ep^p_{W} = \ep_{W^p}$. The distributivity of $\mul$
over $\vee$ yields
$D^{[p]}_W (\mbf{x}) = \bigvee_{k=1}^p\dl^k_{W}(\mbf{x}) =
\bigvee_{k=1}^p(W^k\mul\mbf{x}) = (\bigvee_{k=1}^pW^k)\mul \mbf{x}$
therefore $D^{[p]}_W = \dl_{S_p(W)}$, with
$S_p(W)\ \dot{=} \ \bigvee_{k=1}^p W^k$. Similarly,
$E^{[p]}_W = \ep_{S_p(W)}$. Note that by the same arguments and
Proposition~\ref{prop:adjunction-doa}, we get that $\Doa{n}$ is closed
under $\mul$ and $\vee$.

\section{Graph interpretations}
\label{sec:graph-tropical}

\subsection{Weighted graphs}
\label{sec:weighted-graphs}

Let $W\in\matRmax{n}$ and $\Gcal(W) = (V,E)$ be a weighted and directed
graph containing $n$ vertices whose $n\times n$ adjacency matrix is
$W$, with the convention that $w_{ij}>-\infty$ if and only if $(i,j)\in E$. We now recall that a \emph{path} from vertex $i$ to vertex $j$ in $\Gcal(W)$ is a tuple $\pi = (k_1, \dots, k_l)$ of vertices such that $k_1 = i$, $k_l = j$,
and $(k_p, k_{p+1}) \in E$ for $1 \leq p \leq l-1$. The
\emph{length} of the path, denoted by $\ell(\pi)$, is $l-1$ (the number
of its edges). For $p \geq 1$, $\Gp_{ij}(W)$ denotes the set of paths
from $i$ to $j$ in $\Gcal(W)$ of length $p$ and $\Ginf_{ij}(W)$ the set
of all paths from $i$ to $j$. The
\emph{weight of a path} $\pi = (k_1, \dots, k_l)$, denoted by
$\omega(\pi)$, is the sum
$\omega(\pi) = \sum_{p=1}^{l-1} w_{k_p k_{p+1}}$.

\subsection{Iterated operators}
\label{sec:graph-interpretation-operators}

Recall that for $W\in\matRmax{n}$ and $p\in\N^*$, $W^p$ is the $p$-th
power of $W$ in the $\mul$ sense, and $S_p(W)$ is the matrix defined
in Section~\ref{sec:iterated_operators}, denoted by $S_p$ here for
simplicity. We note respectively $w_{ij}^{(p)}$ and $s^{[p]}_{ij}$
their coefficients. The following result is well known in tropical
algebra and graph theory~\cite{carre71algebra,cuninghame79minimax},
and will help interpret the operators defined earlier. It can be
proved by induction.

\begin{prop}
\label{prop:wp}
Let $W\in \matRmax{n}$ and $p\in\mathbb{N}^*$. Then for
any $1\leq i, j\leq n$,
\begin{equation}
  \label{eq:w_p}
  w_{ij}^{(p)} = \max\left\lbrace \omega(\pi), \ \pi\in\Gp_{ij}(W)\right\rbrace \;\; and \;\; s^{[p]}_{ij} = \max \left\lbrace \omega(\pi), \ \pi\in \bigcup_{1\leq k \leq p}\ \Gamma^{(k)}_{ij}(W)\right\rbrace
\end{equation}
with the convention $\max(\emptyset) = -\infty$.
\end{prop}
The equations in \eqref{eq:w_p} are equivalent to saying that
\begin{enumerate}
\item $w_{ij}^{(p)} > -\infty$ (resp. $s^{[p]}_{ij} > -\infty$) if and
  only if there is at least a path in $\Gcal(W)$ from vertex $i$ to
  vertex $j$ of length exactly (resp. at most) $p$;
\item $w_{ij}^{(p)}$ (resp. $s^{[p]}_{ij}$) is the maximal weight over
  the set of paths from vertex $i$ to vertex $j$ of length exactly
  (resp. at most) $p$.
\end{enumerate}
Therefore the graphs $\Gcal(W^p)$ and $\Gcal(S_p)$ have the same set of
vertices as the original graph $\Gcal(W)$, but an edge exists between
vertices $i$ and $j$ in $\Gcal(W^p)$ (resp. $\Gcal(S_p)$) whenever
there is a path of length exactly (resp. at most) $p$ from $i$ to $j$
in $\Gcal(W)$. The weights associated with this new edge are the
maximal weights over the corresponding set of paths.

Now if $W\in\Doa{n}$, following Section~\ref{sec:iterated_operators}
we get, for $\mbf{x}\in\LatL$ and $i\in\setBracket{1}{n}$:
\begin{equation}
  \label{eq:delta_p}
  \begin{array}{lcr}
   \dl^p_{W} (\mbf{x})_i = \bigvee_{j\in\mathcal{N}^p_i} \lbrace x_j + w^{(p)}_{ij}\rbrace &, \;\;\; & \ep^p_{W} (\mbf{x})_i = \bigwedge_{j\in\check{\mathcal{N}}^p_i} \lbrace x_j - w^{(p)}_{ji}\rbrace
  \end{array}
\end{equation}
and
\begin{equation}
  \label{eq:sup_delta_k}
  \begin{array}{lcr}
    D^{[p]}_W (\mbf{x})_i= \bigvee_{j\in N^p_i} \lbrace x_j + s^{[p]}_{ij}\rbrace   &,\;\;\;& E^{[p]}_W (\mbf{x})_i= \bigwedge_{j\in\check{N}^k_i} \lbrace x_j - s^{[p]}_{ji}\rbrace
  \end{array}
\end{equation}
where $\mathcal{N}^p_i$ is the set of neighbours of vertex $i$ in
$\Gcal(W^p)$ or, equivalently, the set of vertices in $\Gcal(W)$ that
can be reached from $i$ through a path of length $p$;
$\check{\mathcal{N}}^p_i = \Big\lbrace j\in\lbrace 1,\dots, n\rbrace,
i\in \mathcal{N}^p_j\Big\rbrace$;
$N^p_i = \cup_{1\leq k \leq p}\mathcal{N}^k_i$ and
$\check{N}^p_i = \cup_{1\leq k \leq p}\check{\mathcal{N}}^k_i$. Hence
these dilations and erosions are suprema and infima of ``penalised''
values over extended neighbourhoods induced by the original graph. The
penalization is given by the strength of the connection between
vertices: the closer the penalising weight to zero, the more the
neighbours' value contributes to the result. The fact that we can
restrict the supremum and infimum over graph neighbourhoods
in~\eqref{eq:delta_p} and \eqref{eq:sup_delta_k} is due to the weight
values being $-\infty$ outside these neighbourhoods, hence not
contributing to the supremum and infimum.

\subsection{Path interpretation of the opening $G^{[p]}_{W}$}
\label{sec:path_interpretation}

The goal of this section is to show that $G^{[p]}_{W}$ can be
interpreted similarly to a path opening~\cite{heijmans2005path}, in
the sense that it preserves bright values that are connected to other
bright values forming long enough paths in a graph. We can first
remark that for any $\mbf{x}\in\LatL$, $i\in\setBracket{1}{n}$ and
$t\in [a,b]$:
\begin{equation}
  \label{eq:caract_opening_general}
  G^{[p]}_{W}(\mbf{x})_i\geq t \iff \exists j\in N^p_i, \; \text{ such that } \forall l\in
  \check{N}^p_j\;\; x_l \geq t - s^{[p]}_{ij} + s^{[p]}_{lj},
\end{equation}
which is straightforward from the expressions
in~\eqref{eq:sup_delta_k}, as $G^{[p]}_{W} =
D^{[p]}_{W}E^{[p]}_{W}$. This directly yields
\begin{equation}
  \label{eq:opening_general_k}
  G^{[p]}_{W}(\mbf{x})_i = \bigvee\left\lbrace t\in [a,b], \exists j\in N^p_i, \; \forall l\in \check{N}^p_j,\;\; x_l \geq t - s^{[p]}_{ij} + s^{[p]}_{lj}\right\rbrace.
\end{equation}
In the case of binary weights, \emph{i.e.}  $w_{ij} = 0$ if vertex $j$
is neighbour of $i$ in $\Gcal$ and $w_{ij} = -\infty$ otherwise, which
corresponds to a non-weighted graph, then
$s^{[p]}_{ij} = s^{[p]}_{lj} = 0$ in \eqref{eq:caract_opening_general}
and \eqref{eq:opening_general_k}. Therefore, if
$G^{[p]}_{W}(\mbf{x})_i\geq t$, then there is a vertex $j$ which is at
most $p$ steps away from $i$, such that all paths of length at most
$p$ and ending in $j$, including those of length exactly $p$ and
passing through $i$ (if they exist), show values larger than
$t$. In the general case, the additional term
$- s^{[p]}_{ij} + s^{[p]}_{lj}$ modulates this constraint in function of the strength of the connection of $i$ and the other vertices of $\check{N}_j$,   to $j$.

\section{Links to the max-plus spectral theory}
\label{sec:morpho_spectral}

Now we present the consequences and interpretations of some results
from the spectral theory in max-plus algebra. We first report
definitions from~\cite{cuninghame79minimax} necessary to
Theorem~\ref{thm:eigen-problem} (also
from~\cite{cuninghame79minimax}). Then we draw the links to our
setting and more particularly in the case of a symmetric matrix,
corresponding to a non-directed graph. In all this section,
$W\in\matRmax{n}$.

\subsection{General definitions and results}
\label{sec:general_def_results}

\begin{defn}[Eigenvector, eigenvalue \cite{cuninghame79minimax}]
  Let $\mbf{x}\in\Rmax^n$ and $\lambda\in\Rmax$. Then $\mbf{x}$ is an
  \emph{eigenvector} of $W$ with $\lambda$ as corresponding
  $eigenvalue$ if
  $W\mul \mbf{x} = \lambda \mul \mbf{x} = \lambda + \mbf{x}$.
  If there exists finite $\mbf{x}$ and $\lambda$ solutions to this
  equation, we say that the eigenproblem is \emph{finitely soluble}.
\end{defn}

In the graph $\Gcal(W)$, a path $(k_1, \dots, k_l)$ is called a
\emph{circuit} if $k_1 = k_l$. We will note $\Ccal(W)$ the set of all
circuits of $\Gcal(W)$. Circuits allow us to distinguish another class of
matrices in $\matRmax{n}$, called \emph{definite} matrices. They are
important to the present framework as they include the {\doa} matrices.

\begin{defn}[Definite matrix~\cite{cuninghame79minimax}]
  $W$ is said \emph{definite} if $\max_{c\in\Ccal(W)} \omega(c) = 0.$
  In other words, all the circuits of $\Gcal(W)$ have non positive
  weights, and at least one circuit $c^*$, called a zero-weight
  circuit, achieves $\omega(c^*)=0$.
\end{defn}
To see that if $W$ is row or column-$0$-astic, then it is definite, it is sufficient to build an increasing path with zero-weight, until one vertex repeats.
The path can be initialized with any vertex $j_1$.
Then given the current path $(j_1, \dots, j_m)$, we extend it by adding a vertex $j_{m+1}$ such that $w_{j_m j_{m+1}} = 0$. This is always possible thanks to the row or column-0-asticity of $W$.
Since there are $n$ distinct vertices in $\Gcal(W)$, an index will repeat after at most $n$ iterations.

\begin{defn}[Eigen-node, equivalent eigen-nodes~\cite{cuninghame79minimax}]
  Let $W$ be a definite matrix. An \emph{eigen-node} is any vertex in
  $\Gcal(W)$ belonging to a zero-weight circuit. Two eigen-nodes are
  said \emph{equivalent} if there is a zero-weight circuit passing
  through both of them.
\label{def:eigen-node}
\end{defn}
In~\cite{cuninghame79minimax}, $S_n(W) = \bigvee_{1\leq k\leq n}W^k$
is denoted by $\Delta(W)$ and called the $\textbf{metric matrix}$. Recall
that for $i,j\in\setBracket{1}{n}$, $\Delta(W)_{ij}$ is the maximal
weight over the set of paths from vertex $i$ to vertex $j$ of length
at most $n$, in $\Gcal(W)$ (Prop.~\ref{prop:wp}). If $W$ is definite,
circuits have non-positive weights in $\Gcal(W)$ and therefore any
path longer than $n$ can be reduced to a shorter path with non larger
weight. Hence, $\Delta(W)_{ij}$ is actually the maximal weight over
the set of \emph{all} paths from $i$ to $j$. This provides an easy characterisation of eigen-nodes for $W$ definite: $j$ is an
eigen-node of $\Gcal(W)$ if and only if $\Delta(W)_{jj} = 0$. Furthermore, the $j$-th column $\xi_j$ of $\Delta(W)$ is a map
of the ancestors of $j$ in $\Gcal(W)$. It tells which vertices can
reach $j$ and at which cost.
\begin{defn}[Fundamental eigenvectors,
  eigenspace~\cite{cuninghame79minimax}]
  Let $W$ be a definite matrix. Then a fundamental eigenvector of $W$
  is any $j$-th column $\xi_j$ of $\Delta(W)$, where $j$ is an
  eigen-node. Two fundamental eigenvectors are said \emph{equivalent}
  if their associated eigen-nodes are equivalent
  (see Definition~\ref{def:eigen-node}).\\
  Let
  $\mathcal{E} = \lbrace\xi_{i_1}, \xi_{i_2}, \dots, \xi_{i_k}\rbrace$
  be a set of $k\geq 1$ fundamental eigenvectors of $W$, all pairwise
  non-equivalent. The set $\mathcal{E}$ is said to be a maximal set of
  non-equivalent fundamental eigenvectors if any other fundamental
  eigenvector of $W$ is equivalent to one of the eigenvectors in
  $\mathcal{E}$.\\
  In this case the set
  $\lbrace \bigvee_{j=1}^k x_j + \xi_{i_j}, \mbf{x}\in\Rmax^k\rbrace$
  is called the eigenspace of $W$ and does not depend on $\mathcal{E}$
  (see \cite{cuninghame79minimax}, Lemma 24-1).
\end{defn}
\begin{thm}[\cite{cuninghame79minimax}]
  Let $W$ be a {\doa} (hence definite) matrix. Then the following statements are valid: 
  \begin{itemize}
  \item For any fundamental eigenvector $\xi_j$ of $W$ (finite or
    not), $W\mul\xi_j = \xi_j$.
  \item The eigenproblem is finitely soluble.
   \item If two fundamental eigenvectors are equivalent, then they are equal.
  \item Any finite eigenvector is associated to the eigenvalue
    $\lambda = 0$, and lies in the eigenspace of $W$.
  \end{itemize}
  \label{thm:eigen-problem}
\end{thm}

\subsection{Consequences and interpretations}
\label{sec:consequences_spectral_theory}

\paragraph{In general.}
As said, the results of the previous section apply to our setting
since we consider adjunctions represented by {\doa} matricesa. For
$W\in\Doa{n}$, $\Delta(W)$ is also in $\Doa{n}$ and the corresponding
opening $\dl_{\Delta(W)}\ep_{\Delta(W)}$ is $G^{[n]}_W$. By
definition, $G^{[n]}_W(\mbf{x})$ projects $\mbf{x}\in\LatL$ onto
$\dl_{\Delta(W)}(\LatL)$, which is the set
$\lbrace \bigvee_{j=1}^n y_j + \xi_{j}, \mbf{y}\in\LatL \rbrace$ of
max-plus combinations of columns of
$\Delta(W)$. Theorem~\ref{thm:eigen-problem} tells that this
decomposition can be split as
$G^{[n]}_W(\mbf{x}) = \mbf{u}\vee\mbf{v}$, where $\mbf{u}$ lies in the
eigenspace of $W$ and $\mbf{v}$ is a max-plus combination of the
$\xi_j$s which are not fundamental eigenvectors. This decomposition
may be sparser than the original one, as the dimension of the
eigenspace of $W$, \emph{i.e.} $Card(\mathcal{E}$), can be lower than
the number of fundamental eigenvectors.

\paragraph{The case of symmetric $W\in\Doa{n}$.}
This case corresponds to considering a
non-directed graph supporting the signal $\mbf{x}$. As the adjacency
relationship is often based on a symmetrical function on pairs of
vertex values, this assumption covers many practical cases
(e.g.~\cite{blusseau18tropical,blusseau2022anisotropic}).
The main consequence of $W\in\Doa{n}$ symmetric is that \emph{every
  vertex} $j$ is an eigen-node: for any $j\in\setBracket{1}{n}$ there
is~$i$ such that $w_{ij} = 0 = w_{ji}$ and therefore $(j,i,j)$ is a zero-weight circuit. This entails three other consequences. 

First,
$\Delta(W)_{jj} = 0$ for every $j\in\setBracket{1}{n}$, following the
characterisation of eigen-nodes described earlier, which implies that
$\dl_{\Delta(W)} = D^{[n]}_W$ is extensive and
$\ep_{\Delta(W)} = E^{[n]}_W$ anti-extensive
(Prop.~\ref{prop:adjunction-doa}). Secondly, $W\mul\xi_j = \xi_j$ for
every column $\xi_j$ of $\Delta(W)$, which implies
$W^k\mul\xi_j = \xi_j$ for $1\leq k\leq n$, hence
$\Delta(W)\mul\xi_j = \xi_j$ and finally
$\Delta(W)\mul\Delta(W) = \Delta(W)$. This means $D^{[n]}_W$ and
$E^{[n]}_W$ are idempotent. They are therefore a closing and an
opening respectively and $E^{[n]}_W = G^{[n]}_W$, since an
adjunction $(\ep, \ \dl)$ for which $\ep$ is an opening and $\dl$ a
closing verifies $\ep = \dl\ep$ (and $\dl = \ep\dl$). The third
consequence is the following.
\begin{corol}
If $W\in\Doa{n}$ is symmetric, then the set of invariants of $G^{[n]}_W$ is exactly the eigenspace of $W$.
\end{corol}
When $W$ is symmetric, a maximal set of $k$ non-equivalent fundamental eigenvectors $\lbrace\xi_{i_1}, \xi_{i_2}, \dots, \xi_{i_k}\rbrace$, $k\leq n$, can be seen as negative distance maps to the $k$ corresponding eigen-nodes $\Gcal(W)$, as they contain the optimal cost (maximal weight) between any vertex and the eigen-nodes\footnote{Note that $\Delta(W)$ is a metric, not exactly between vertices, but between their equivalence classes induced by Def.~\ref{def:eigen-node}, as all vertices are eigen-nodes when $W$ is symmetric.}. Hence we can picture the aspect of $G^{[n]}_W(\mbf{x})$, for $\mbf{x}\in\LatL$: it is the upper-envelope of the largest vertical translations of these distance maps that are dominated by $\mbf{x}$. Therefore, adapting $\Gcal(W)$ to $\mbf{x}$ by well connecting vertices within relevant structures preserves these structures under the filter $G^{[n]}_W$, as shown in~\cite{blusseau18tropical,blusseau2022anisotropic}. In practice, $n$ might be large, such as the number of pixels of an image. Since $(G^{[p]}_W )_{1\leq p\leq n}$ is a granulometry, we know that $G^{[n]}_W$ can be approximated by $G^{[p]}_W$ with increasing $p$.

\section{Conclusion}
\label{sec:conclusion}
In this paper we consolidated the basis of the representation of
adjunctions by matrices in max-plus algebra. We showed that it is a very
flexible framework that generalises many types of morphological
adjunctions. In particular, it allows describing precisely the
behaviour of iterated operators based on spatially-variant, non-flat
structuring functions. This is made possible by their graph
interpretation and spectral results in max-plus algebra. Future works
shall investigate further the insights that max-plus algebra can bring
to mathematical morphology through this framework.

\bibliographystyle{splncs04}
\bibliography{max_plus}

\begin{thebibliography}{10}
\providecommand{\url}[1]{\texttt{#1}}
\providecommand{\urlprefix}{URL }
\providecommand{\doi}[1]{https://doi.org/#1}

\bibitem{akian06max-plus-algebra}
Akian, M., Bapat, R., Gaubert, S.: Max-plus algebra. Handbook of linear algebra
  ({D}iscrete {M}athematics and its {A}pplications)  \textbf{39},  10--14
  (2006)

\bibitem{blusseau18tropical}
Blusseau, S., Velasco-Forero, S., Angulo, J., Bloch, I.: Tropical and
  morphological operators for signal processing on graphs. In: 25th IEEE
  International Conference on Image Processing. pp. 1198--1202 (2018)

\bibitem{blusseau2022anisotropic}
Blusseau, S., Velasco-Forero, S., Angulo, J., Bloch, I.: {Adaptive Anisotropic
  Morphological Filtering Based on Co-Circularity of Local Orientations}.
  {Image Processing On Line}  \textbf{12},  111--141 (2022),
  \url{https://doi.org/10.5201/ipol.2022.397}

\bibitem{bouaynaya2008theoretical}
Bouaynaya, N., Charif-Chefchaouni, M., Schonfeld, D.: Theoretical foundations
  of spatially-variant mathematical morphology {P}art {I}: Binary images. IEEE
  Transactions on Pattern Analysis and Machine Intelligence  \textbf{30}(5),
  823--836 (2008)

\bibitem{carre71algebra}
Carr{\'e}, B.A.: An algebra for network routing problems. IMA Journal of
  Applied Mathematics  \textbf{7}(3),  273--294 (1971)

\bibitem{cuninghame79minimax}
Cuninghame-Green, R.A.: Minimax algebra, vol.~166. Springer-Verlag Berlin
  Heidelberg (1979)

\bibitem{debayle2006general}
Debayle, J., Pinoli, J.C.: General adaptive neighborhood image processing: Part
  {I}: Introduction and theoretical aspects. Journal of Mathematical Imaging
  and Vision  \textbf{25}(2),  245--266 (2006)

\bibitem{heijmans2005path}
Heijmans, H., Buckley, M., Talbot, H.: Path openings and closings. Journal of
  Mathematical Imaging and Vision  \textbf{22}(2),  107--119 (2005)

\bibitem{heijmans1990algebraic}
Heijmans, H., Ronse, C.: The algebraic basis of mathematical morphology {I}.
  dilations and erosions. Computer Vision, Graphics, and Image Processing
  \textbf{50}(3),  245--295 (1990)

\bibitem{lerallut2007amoebas}
Lerallut, R., Decenci\`ere, E., Meyer, F.: Image filtering using morphological
  amoebas. Image and Vision Computing  \textbf{25}(4),  395--404 (2007)

\bibitem{maragos13representations}
Maragos, P.: Chapter {T}wo - {R}epresentations for morphological image
  operators and analogies with linear operators. Advances in Imaging and
  Electron Physics, vol.~177, pp. 45 -- 187. Elsevier (2013)

\bibitem{ronse1990why}
Ronse, C.: Why mathematical morphology needs complete lattices. Signal
  Processing  \textbf{21}(2),  129--154 (1990)

\bibitem{salembier09study}
Salembier, P.: Study on nonlocal morphological operators. In: 16th IEEE
  International Conference on Image Processing. pp. 2269--2272 (2009)

\bibitem{velasco13nonlocal}
Velasco-Forero, S., Angulo, J.: On nonlocal mathematical morphology. In:
  International Symposium on Mathematical Morphology and Its Applications to
  Signal and Image Processing. pp. 219--230. Springer (2013)

\bibitem{velasco-forero15nonlinear}
Velasco-Forero, S., Angulo, J.: Nonlinear operators on graphs via stacks. In:
  Nielsen, F., Barbaresco, F. (eds.) Geometric Science of Information,
  Proceedings. pp. 654--663. Springer International Publishing (2015)

\bibitem{verdu2011anisotropic}
Verd\'u-Monedero, R., Angulo, J., Serra, J.: Anisotropic morphological filters
  with spatially-variant structuring elements based on image-dependent gradient
  fields. IEEE Transactions on Image Processing  \textbf{20}(1),  200--212
  (2011)

\end{thebibliography}

\end{document}